\documentclass[journal]{new-aiaa} 
\usepackage[utf8]{inputenc}
\usepackage{textcomp}

\usepackage{graphicx}
\usepackage{amsmath}
\usepackage[version=4]{mhchem}
\usepackage{siunitx}

\usepackage{bm}
\usepackage{amsmath,amsfonts}
\usepackage{xcolor}

\DeclareMathOperator{\R}{\mathbb{R}}

\newtheorem{thm}{Theorem}[section]
\newtheorem{lem}[thm]{Lemma}

\newtheorem{defn}[thm]{Definition}
\newtheorem{rem}[thm]{Remark}
\newtheorem{assum}[thm]{Assumption}
\newenvironment{proof}{\paragraph{Proof:}}{\hfill$\square$}
\usepackage{lipsum}

\usepackage{algorithm}
\usepackage{algorithmic}
\floatname{algorithm}{Procedure}

\title{Consensus Control of Linear Multi-Agent Systems with Non-uniform Time-varying Communication Delays}
\author{Rajnish Bhusal \footnote{Ph.D. Student, Department of Mechanical and Aerospace Engineering,  \texttt{rajnish.bhusal@mavs.uta.edu}} and Kamesh Subbarao \footnote{Professor, Department of Mechanical and Aerospace Engineering,  \texttt{subbarao@uta.edu}}}
\affil{The University of Texas at Arlington, Arlington, TX, 76019, USA }

\begin{document}
\maketitle


\begin{abstract}
This paper is concerned with the consensus problem for multi-agent systems subject to communication delays between the neighboring agents. We consider a scenario where each agent is characterized by a general high-order linear system and the communication delays between the agents are non-uniform and time-varying. We design a distributed control protocol for the agents and provide an equivalent stability problem to be solved that guarantees the state consensus in the group of agents. Moreover, a delay-dependent stability criterion is provided by combining the Lyapunov–Krasovskii method with the linear matrix inequality approach.
\end{abstract}

\section{Introduction}
\label{sec:intro}
Consensus problems in multi-agent systems have received great attention in recent years because of their wide applications in the areas including vehicle formations \cite{fax2004information}, distributed computation and optimization \cite{yang2016multi}, flocking \cite{jafari2019biologically}, and distributed smart grids \cite{ma2013smart}. The objective of the consensus problem is to guarantee agreement on a common state or output trajectory while designing control protocols for each agent. 

In practice, agents share information through wireless communication networks.  Due to technological constraints (communication congestion and finite transmission bandwidth), time-delays in the information transmission and processing between agents are unavoidable \cite{hou2017consensus}. Further, the presence of time delay may result in degraded closed-loop performance, and may even cause the loss of stability. Thus, the effects of time delays on the closed-loop network of multi-agent system need to be investigated. In the literature, two types of time delays are considered in the scenario of multi-agent systems namely, input delay and communication delay \cite{cao2013overview}. Input delay relates to the processing and connecting time for the packets arriving at each agent while communication delay refers to the time for transferring information between agents.  

In the literature, the consensus problem in first-order, second-order and general linear multi-agent systems with communication delays and/or input delays has been solved using frequency-domain and time-domain (Lyapunov)-based approaches. For first-order multi-agent systems under undirected graph topology, Ref. \cite{ma2020delay} derives the analytical expression for the input delay margin by solving a univariate convex optimization problem. Similarly, the delay margins for a general second-order and double integrator multi-agent systems with communication delay under directed graph topology are deduced in Ref. \cite{hou2017consensus} using frequency-domain analysis. However, all of these works assume the time delay in the multi-agent systems to be uniform (or identical) for all the agents, which is generally too-restrictive for the real applications. Departing from these works, Ref. \cite{li2018consensus} derives the consensus conditions for a general high-order linear multi-agent systems under undirected graph and non-uniform communication delays. Similarly, the works carried out in \cite{zhang2015synchronization} characterizes the non-uniform input delay margin for a discrete-time high-order linear multi-agent systems under undirected graph topology. 

In the real-time communication processes, time delays are not constant and usually change over time; however, most of the works carried out in the literature of the multi-agent systems, including the aforementioned works, do not take the time-varying nature of delays into consideration. In \cite{wang2020consensus}, a consensus condition for a general high-order discrete-time linear multi-agent systems under undirected graph with time-varying delays is derived, however the delays are assumed to be identical (uniform) for all the agents in the network. For linear multi-agent systems with time-varying non-uniform communication delays under directed graph topology, Ref. \cite{petrillo2017adaptive} provides a delay-dependent stability criterion by combining the Lyapunov-Krasovskii method with the linear matrix inequality (LMI) approach. However, the works carried out in \cite{petrillo2017adaptive} assumes the system matrices to be in a controllable-canonical form. 

To our concern, most of the works carried out in the literature for linear multi-agent systems assume the time-delay to be uniform, and/or constant. On the other hand, the assumptions on the dynamics of the system and underlying graph topology are restrictive. To that end, this paper provides following contribution to the literature. We consider a general linear high-order continuous-time multi-agent systems with non-uniform time-varying communication delays among the agents. The consensus problem in multi-agent systems is converted to an equivalent problem of stabilizing global dynamics of the consensus error. In order to derive the stability conditions for the dynamics of the consensus error, a suitable Lyapunov Krasovskii functional is chosen. Finally, we utilize integral inequalities provided in \cite{han2005absolute} and \cite{peng2008improved}, and invoke Lyapunov Krasovskii stability theorem to derive the delay-dependent stability conditions in terms of linear matrix inequalities.

The paper is organized as follows. Section \ref{sec:primer} describes the preliminaries and formulates the problem in consideration. Section \ref{sec:consensus_protocol} discusses the control protocol for multi-agent systems with non-uniform time-varying communication delays. The main results concerning the delay-margin characterization for the multi-agent system are presented in section \ref{sec:main_results}. Numerical examples are presented in section \ref{sec:results} and section \ref{sec:conclusion} provides the concluding remarks.

\section{Preliminaries and Problem Formulation}
\label{sec:primer}
\subsection{Notations}
For a vector $\mathbf{x} \in \R^n$, $\|\mathbf{x}\|$ denotes its 2-norm. In the paper, $\mathbf{A} \otimes \mathbf{B}$ denotes the Kronecker product of matrices $\mathbf{A}$ and $\mathbf{B}$, $\mathbf{1}_n$ denotes a $n$-dimensional vector of ones; $\mathbf{0}_n$ denotes a $n$-dimensional vector of zeros; $\mathbf{I}_n$ denotes the identity matrix of dimension $n \times n$. Denote $\text{col}(\mathbf{x}_1, \mathbf{x}_2, \dots, \mathbf{x}_n)$ as concatenation of vectors $\mathbf{x}_1, \mathbf{x}_2, \dots, \mathbf{x}_n$ such that $\text{col}(\mathbf{x}_1, \mathbf{x}_2, \dots, \mathbf{x}_n) = [\mathbf{x}^\text{T}_1, \mathbf{x}^\text{T}_2, \dots, \mathbf{x}^\text{T}_n ]^\text{T}$.

\subsection{Preliminaries}
The interconnection among a group of $N$ agents are encoded through communication graphs $\mathcal{G} = (\mathcal{V},\mathcal{E})$ where $\mathcal{V} = \left\lbrace 1, 2, \dots, N \right\rbrace$ is a nonempty node set and $\mathcal{E} \subseteq \mathcal{V} \times \mathcal{V}$ is an edge set of ordered pairs of nodes, called edges. Each of the edges of a graph $(i,j)$ is associated with a non-negative weight $a_{ij}$.  Node $j$ is the neighbor of $i$ if  $(j, i) \in \mathcal{E}$ and the set of neighbors of node $i$ can be represented as $\mathcal{N}_i$. The graph $\mathcal{G}$ is said to be strongly connected if $i, j$ are connected for all distinct nodes $i, j \in \mathcal{V}$.  A graph is said to have a directed spanning tree if there exists a node called the root node, which has no parent node and has directed paths to all other nodes in the graph.

The adjacency matrix $\mathcal{A} = [a_{ij}] \in \R^{N \times N}$ of a directed graph is defined such that  $a_{ij} = 1$ if $(j, i$) $\in$ $\mathcal{E}$ and $a_{ij} = 0$, otherwise. The in-degree of node $v_i$ is defined as $d_i = \sum_{j=1}^n a_{ij}$. The diagonal matrix obtained from $d_i$ as diagonal entries is called diagonal in-degree matrix ($\mathcal{D}$). Finally, the graph Laplacian matrix is obtained as $\mathbf{L} = \mathcal{D} -\mathcal{A} \in \R^{N \times N}$. 

\begin{assum} \label{connected_root_node}
Throughout the paper, the graph is assumed to be strongly connected with atleast one directed spanning tree.
\end{assum}

\begin{rem}
Consider a strongly connected graph with a spanning tree comprising of $N$ nodes. Let $\lambda_i$, $i=1,2, \dots, N$ be the eigenvalues of the Laplacian matrix. Then, $\lambda_1=0$ is always a simple and the smallest eigenvalue of the Laplacian matrix, and $\lambda_j>0$, for all $j=2, \dots, N$ \cite{ren2007information}.
\end{rem}

\begin{lem} \cite{han2005absolute}
\label{lem_S}
For any constant matrix $\mathbb{X} \in \mathbb{R}^{n \times n}$, $\mathbb{X}=\mathbb{X}^{\text{T}}>0$,  a scalar $\gamma>0$, and a vector function $\dot{\mathbf{z}}:[-\gamma, 0] \to \R^{n}$, following integral inequality holds

\begin{equation}
-\gamma \int_{t-\gamma}^{t} \dot{\mathbf{z}}(t+\theta)^{\text{T}} \mathbb{X} \dot{\mathbf{z}}(t + \theta) d \theta \ \leq \ \begin{bmatrix}
\mathbf{z}(t) \\
\mathbf{z}(t-\gamma)
\end{bmatrix}^{\text{T}} \begin{bmatrix}
-\mathbb{X} & \mathbb{X} \\
* & -\mathbb{X}
\end{bmatrix}\begin{bmatrix}
\mathbf{z}(t) \\
\mathbf{z}(t-\gamma)
\end{bmatrix}
\end{equation}
\end{lem}

\begin{lem} \cite{peng2008improved}
\label{lem_R}
For any constant matrix $\mathbb{Y} \in \R^{n \times n}$, $\mathbb{Y} = \mathbb{Y}^{\text{T}}>0$, scalars $h_{1} \leq \tau(t) \leq h_{2}$ and a vector function $\dot{\mathbf{z}}:\left[-h_{2},-h_{1}\right] \to \R^{n}$, following integral inequality holds
\begin{equation}
-\left(h_{2}-h_{1}\right) \int_{t-h_{2}}^{t-h_{1}} \dot{\mathbf{z}}^{\text{T}}(s) \ \mathbb{Y} \ \dot{\mathbf{z}}(s) \ d s \leq \begin{bmatrix}
\mathbf{z}\left(t-h_{1}\right) \\
\mathbf{z}(t-\tau(t)) \\
\mathbf{z}\left(t-h_{2}\right)
\end{bmatrix}^{\text{T}} \begin{bmatrix}
-\mathbb{Y} & \mathbb{Y} & \mathbf{0} \\
* & -2 \mathbb{Y} & \mathbb{Y} \\
* & * & \mathbb{Y}
\end{bmatrix} \begin{bmatrix}
\mathbf{z}\left(t-h_{1}\right) \\
\mathbf{z}(t-\tau(t)) \\
\mathbf{z}\left(t-h_{2}\right)
\end{bmatrix}
\end{equation}
\end{lem}

\subsection{Problem Formulation}
\label{sec:mas_gen}
Let us consider a multi-agent system comprising of $N$ agents with identical dynamics. The dynamics of the $i$th agent is considered to be as follows:
\begin{equation}
\dot{\mathbf{x}}_i(t) = \mathbf{A} \mathbf{x}_i(t) + \mathbf{B} \mathbf{u}_i(t), \qquad i=1,\dots, N \label{single_agent_dyn}
\end{equation}
where $\mathbf{A} \in \R^{n \times n}$, $\mathbf{B} \in \R^{n \times m}$ are the system matrices with $\mathbf{x}_i \in \R^n$ as the state and $\mathbf{u}_i \in \R^m$ as the input of the $i$th agent. 

\begin{assum} \label{assum_stab}
$(\mathbf{A}, \mathbf{B})$ is stabilizable.
\end{assum}

In this paper, we design a fully distributed control protocol for the multi-agent system in \eqref{single_agent_dyn} to achieve consensus under non-uniform time-varying communication delays among the agents. 

\section{Distributed Control Protocol}
\label{sec:consensus_protocol}
\begin{defn} 
\label{consensus_defn}
The group of agents are said to reach consensus under any control protocol $\mathbf{u}_i$ if for any set of initial conditions $\{\mathbf{x}_i(0)\}$ there exists $\mathbf{x}^c\in \R^n$ such that $\lim_{t \to \infty} \mathbf{x}_i(t)=\mathbf{x}^c$ for all $i=1, \dots, N$.
\end{defn}

With assumption \eqref{assum_stab}, let each of the agents $i=1, \dots, N$ have identical feedback controller $K \in \R^{m \times n}$ such that $\mathbf{A}-\mathbf{B} \mathbf{K}$ is Hurwitz. We consider following distributed control protocol based on the relative states between neighboring agents:
\begin{equation}
\mathbf{u}_i(t) = \mathbf{K} \left[ \sum_{j \in \mathcal{N}_i} a_{ij} (\mathbf{x}_i (t-\tau_{ij}(t)) - \mathbf{x}_j(t-\tau_{ij}(t)) \right] \label{control_protoc_delay_nonuniform}
\end{equation}
where $\tau_{ij}(t)$ is the time-delay in the communication between agents $i$ and $j$.

We assume the time delay $\tau_{ij}(t)$ and the delay derivatives $\dot{\tau}_{ij}(t)$  to be bounded for all $i,j =1, \dots, N$ and for all $t>0$ such that $\tau_{ij}(t) \leq \bar{\tau}_{ij}$, and $\dot{\tau}_{ij}(t) \leq \mu_{ij}<1$. Moreover, we consider that the communication delay is dependent on the direction of the information flow, i.e, $\tau_{ij} \neq \tau_{ji}$. Therefore, a unique time-delay is associated with each edge in the graph. Let $r\leq N(N-1)$ be the total number of edges in the graph and $\tau_k$, $k=1, \dots r$ be the delay associated with the $k$th edge. Let $\mathbf{L}_k \in \R^N$ be the Laplacian matrix of the subgraph associated with the time delay $\tau_{k}$ such that $\mathbf{L} = \sum^r_{k=1} \mathbf{L}_k$. For the ease of analysis, we provide the rule to compute $k$ in Procedure \ref{single_indexing_delay} (Appendix).

With the distributed control protocol in \eqref{control_protoc_delay_nonuniform}, the closed-loop dynamics of agent $i$, for all $i=1,\dots, N$  can be written as
\begin{equation}
\dot{\mathbf{x}}_i(t) = \mathbf{A} \mathbf{x}_i(t) + \mathbf{B} \mathbf{K} \left[ \sum_{j \in \mathcal{N}_i} a_{ij} (\mathbf{x}_i (t-\tau_{ij}(t)) - \mathbf{x}_j(t-\tau_{ij}(t)) \right] \label{closed_loop_withdelay_agent_i}
\end{equation}

Denote $\mathbf{x} = \text{col}(\mathbf{x}_1, \mathbf{x}_2, \dots, \mathbf{x}_N) \in \R^{Nn}$ as the global state vector. Now the global state dynamics can be written as
\begin{equation}
\dot{\mathbf{x}}(t)  =  \left(\mathbf{I}_N \otimes \mathbf{A}\right) \mathbf{x}(t) + \sum^r_{k=1} \left(\mathbf{L}_k  \otimes  \mathbf{B}  \mathbf{K} \right) \mathbf{x}(t-\tau_k) \label{global_state_dyn_wdelay}
\end{equation}

\section{Main Results}
This section provides the delay-dependent stability conditions for the consensus in multi-agent system in \ref{global_state_dyn_wdelay}. In order to derive the stability conditions, the consensus problem is transformed to an equivalent stability problem using the following Lemma.

\label{sec:main_results}
\begin{lem}
\label{lem_z_x}
Suppose the graph topology of $N$ agents satisfy Assumption \ref{connected_root_node}. Then, the multi-agent system in \eqref{global_state_dyn_wdelay} reaches consensus asymptotically (as per the definition \ref{consensus_defn}), if the following equivalent system is asymptotically stable
\begin{equation}
\dot{\mathbf{z}}(t) = \left(\mathbf{I}_{N-1} \otimes \mathbf{A} \right) \mathbf{z}(t) - \sum_{k=1}^{r} \left(\bar{\mathbf{L}}_k \otimes \mathbf{B} \mathbf{K}  \right) \mathbf{z}\left(t-\tau_{k}(t)\right)  \label{z_dynamics}
\end{equation}
where $\mathbf{z} = \text{col} (\mathbf{z}_1, \mathbf{z}_2, \dots, \mathbf{z}_{N-1}) \in \R^{(N-1) n}$ is the global consensus error such that $\mathbf{z}_i = \mathbf{x}_1-\mathbf{x}_{i+1} \in \R^n$, ($i=1, \cdots, N-1)$. Moreover, $\bar{\mathbf{L}}_k = \mathbf{U} \mathbf{L}_k \mathbf{W} \in \R^{(N-1) \times (N-1)}$ with $\mathbf{U} = \begin{bmatrix}
\mathbf{1}_{N-1} & -\mathbf{I}_{N-1} \end{bmatrix} \in \R^{(N-1) \times N}$, and $\mathbf{W} = \begin{bmatrix}
\mathbf{0}^{\text{T}}_{N-1} \\ -\mathbf{I}_{N-1} 
\end{bmatrix} \in \R^{N \times (N-1)}
$.
\end{lem}

\begin{proof}
Please refer to the Appendix.
\end{proof}

In order to derive the stability conditions of the consensus error system in \eqref{z_dynamics}, we introduce Theorem \ref{thm_main}. For notational simplicity, we denote identity matrix $\mathbf{I}_{N-1}$ with $\mathbf{I}$. 

\begin{thm}
\label{thm_main}
Suppose the graph topology of $N$ agents satisfy Assumption \ref{connected_root_node}. Then, the multi-agent system in \eqref{global_state_dyn_wdelay} reaches consensus asymptotically (as per the definition \ref{consensus_defn}) for any $\tau_{k}(t)$ satisfying $\tau_{k}(t) \leq \bar{\tau}_{k}$, and $\dot{\tau}_{k}(t) \leq \mu_{k}<1$ ($k=1, \dots, r$) if there exists $\mathbf{P}=\mathbf{P}^{\text{T}}>0 \in \R^{n \times n}$, $\mathbf{Q}_k=\mathbf{Q}^{\text{T}}_k>0 \in \R^{n \times n}$, $\mathbf{R}_k = \mathbf{R}^{\text{T}}_k >0 \in \R^{n \times n}$, $k=1, \dots, r$ and $\mathbf{S}_{kj}=\mathbf{S}^{\text{T}}_{kj}>0 \in \R^{n \times n}$, $k=1, \dots, r-1$, $j=i+1, i+2, \dots, r$ such that the following LMI holds
\begin{equation}
\begin{bmatrix}
\mathbf{\Pi} & \bm{\xi}^{T} \mathbf{\Gamma} \\
* & -\mathbf{\Gamma}
\end{bmatrix}
<\mathbf{0} \label{lmi_delay}
\end{equation}
where
\begin{equation}
\begin{aligned}
\mathbf{\Gamma} &= \mathbf{I} \otimes  \left(\sum_{k=1}^{r} \bar{\tau}_{k}^{2} \mathbf{R}_k +\sum_{k=1}^{r-1} \sum_{j=k+1}^{r}\left(\bar{\tau}_{k}-\bar{\tau}_{j}\right)^{2} \mathbf{S}_{k j} \right)\\
\end{aligned},
\end{equation}
\begin{equation}
\bm{\xi} = \begin{bmatrix}
\left( \mathbf{I} \otimes \mathbf{A} \right) & -\left(\bar{\mathbf{L}}_1 \otimes \mathbf{B} \mathbf{K}  \right) & \mathbf{0} & -\left(\bar{\mathbf{L}}_2 \otimes \mathbf{B} \mathbf{K}  \right) & \mathbf{0} & \cdots & -\left(\bar{\mathbf{L}}_r \otimes \mathbf{B} \mathbf{K}  \right) & \mathbf{0}
\end{bmatrix}, \label{gamma_mat}
\end{equation} and
\begin{equation}
\mathbf{\Pi} = \begin{bmatrix}
\mathbf{\Pi}_{1,1} & \mathbf{\Pi}_{1,2}& \mathbf{0} & \mathbf{\Pi}_{1,4}& \mathbf{0} & \cdots &\mathbf{\Pi}_{1,2r} & \mathbf{0} \\
* & \mathbf{\Pi}_{2,2} &  \mathbf{I} \otimes \mathbf{R}_{1} &  \mathbf{I} \otimes \mathbf{S}_{12} & \mathbf{0} & \cdots &  \mathbf{I} \otimes \mathbf{S}_{1r} & \mathbf{0} \\
* & * & - \mathbf{I} \otimes \mathbf{R}_{1} & \mathbf{0} & \mathbf{0} & \cdots & \mathbf{0} & \mathbf{0}\\
* & * & * & \mathbf{\Pi}_{4,4} &  \mathbf{I} \otimes \mathbf{R}_{2} & \cdots &  \mathbf{I} \otimes \mathbf{S}_{2r} & \mathbf{0}\\
* & * & * & * & - \mathbf{I} \otimes \mathbf{R}_{2} & \cdots & \mathbf{0} & \mathbf{0}\\
\vdots & \vdots & \vdots & \vdots & \vdots & \vdots & \vdots & \vdots \\
* & * & * & * & * & \cdots & \mathbf{\Pi}_{2k, 2k} & \mathbf{I} \otimes \mathbf{R}_{r}\\
* & * & * & * & * & \cdots &* &- \mathbf{I} \otimes \mathbf{R}_{r}
\end{bmatrix}
\label{pi_mat}
\end{equation}
with
\begin{equation}
\begin{aligned}
\mathbf{\Pi}_{11}  =\ & \mathbf{I} \otimes \left[ \mathbf{A}^{\text{T}} \mathbf{P} + \mathbf{P} \mathbf{A} + \sum_{k=1}^{r} \left(\mathbf{Q}_k - \mathbf{R}_k \right) \right], \quad \mathbf{\Pi}_{1, 2k} =  \left(\mathbf{I} \otimes \mathbf{R}_{k} \right) - \left(\bar{\mathbf{L}}_k \otimes \mathbf{P} \mathbf{B} \mathbf{K} \right)\\
 \mathbf{\Pi}_{2 k, 2 k} =\ &  \mathbf{I} \otimes \left(-\left(1-\mu_{k}\right) \mathbf{Q}_k - 2 \mathbf{R}_k - \sum_{j=1}^{k-1} \mathbf{S}_{j k} - \sum_{j=k+1}^{r} \mathbf{S}_{k j} \right).
\end{aligned}
\end{equation}

\end{thm}

\begin{proof}
Let us consider the following Lyapunov Krasovskii functional
\begin{equation}
V\left(\mathbf{z}_{t}, t\right)= V_1 \left(\mathbf{z}_{t}, t\right) + V_2 \left(\mathbf{z}_{t}, t\right) + V_3 \left(\mathbf{z}_{t}, t\right) + V_4\left(\mathbf{z}_{t}, t\right)
\end{equation}
where
\begin{equation}
\begin{aligned}
V_1\left(\mathbf{z}_{t}, t\right)=\ & \mathbf{z}^{\text{T}}(t) \left(\mathbf{I} \otimes \mathbf{P} \right) \mathbf{z}(t)\\
V_2\left(\mathbf{z}_{t}, t\right)=\ & \sum_{k=1}^{r}\left[\int_{t-\tau_{k}(t)}^{t} \mathbf{z}^{\text{T}}(s) \left(\mathbf{I} \otimes \mathbf{Q}_k \right) \mathbf{z}(s) d s\right] \\
V_3\left(\mathbf{z}_{t}, t\right)=\ & \sum_{k=1}^{r}\left[\bar{\tau}_{k} \int_{-\bar{\tau}_{k}}^{0} \int_{t+\theta}^{t} \dot{\mathbf{z}}^{\text{T}}(s) \left(\mathbf{I} \otimes \mathbf{R}_k \right) \dot{\mathbf{z}}(s) d s d \theta \right] \\
V_4 \left(\mathbf{z}_{t}, t\right)=\ & \sum_{k=1}^{r-1} \sum_{j=k+1}^{r}\left[\left(\bar{\tau}_{k}-\bar{\tau}_{j}\right) \int_{-\bar{\tau}_{k}}^{-\bar{\tau}_{j}} \int_{t+\theta}^{t} \dot{\mathbf{z}}^{\text{T}}(s) \left(\mathbf{I} \otimes \mathbf{S}_{k j} \right) \dot{\mathbf{z}}(s) d s d \theta \right]
\end{aligned}
\end{equation}

The time-derivative of the components of the Lyapunov-Krasovskii functional along the trajectory of \eqref{z_dynamics} is given by
\begin{equation}
\label{lyap_components}
\begin{aligned}
\dot{V}_1 \left(\mathbf{z}_{t}, t\right)=\ & 2 \mathbf{z}^{\text{T}}(t) \left(\mathbf{I} \otimes \mathbf{P} \right)\left[\left( \mathbf{I} \otimes \mathbf{A} \right) \mathbf{z}(t)-\sum_{k=1}^{r} \left(\bar{\mathbf{L}}_k \otimes \mathbf{B} \mathbf{K}  \right) \mathbf{z}\left(t-\tau_{k}(t)\right)\right]\\
\dot{V}_2 \left(\mathbf{z}_{t}, t\right)=\ &  \sum_{k=1}^{r}\left[\mathbf{z}^{\text{T}}(t) \left(\mathbf{I} \otimes \mathbf{Q}_k \right) \mathbf{z}(t)-\left(1-\dot{\tau}_{k}(t)\right) \mathbf{z}^{\text{T}}\left(t-\tau_{k}(t)\right) \left(\mathbf{I} \otimes \mathbf{Q}_k \right) \mathbf{z}\left(t-\tau_{k}(t)\right)\right] \\
\dot{V}_3 \left(\mathbf{z}_{t}, t\right)=\ & \sum_{k=1}^{r}\left[\bar{\tau}_{k} \int_{t-\bar{\tau}_{k}}^{t} \dot{\mathbf{z}}^{\text{T}}(t) \left(\mathbf{I} \otimes \mathbf{R}_k \right) \dot{\mathbf{z}}(t) d \theta \right]  - \sum_{k=1}^{r}\left[\bar{\tau}_{k} \int_{t-\bar{\tau}_{k}}^{t} \dot{\mathbf{z}}^{\text{T}}(t+\theta) \left(\mathbf{I} \otimes \mathbf{R}_k \right) \dot{\mathbf{z}}(t+\theta) d \theta\right]\\
= \ &  \dot{\mathbf{z}}^{\text{T}}(t) \left[ \sum_{k=1}^{r}\bar{\tau}^2_{k} \left(\mathbf{I} \otimes \mathbf{R}_k \right)  \right]  \dot{\mathbf{z}}(t) - \sum_{k=1}^{r}\left[\bar{\tau}_{k} \int_{t-\bar{\tau}_{k}}^{t} \dot{\mathbf{z}}^{\text{T}}(t+\theta) \left(\mathbf{I} \otimes \mathbf{R}_k \right) \dot{\mathbf{z}}(t+\theta) d \theta\right]\\
\dot{V}_4 \left(\mathbf{z}_{t}, t\right)=\ &  \sum_{k=1}^{r-1} \sum_{j=k+1}^{r}\left[\left(\bar{\tau}_{k}-\bar{\tau}_{j}\right) \int_{t-\bar{\tau}_{k}}^{t-\bar{\tau}_{j}} \dot{\mathbf{z}}^{\text{T}}(t) \left(\mathbf{I} \otimes \mathbf{S}_{k j} \right) \dot{\mathbf{z}}(t) d \theta \right] \\
& -\sum_{k=1}^{r-1} \sum_{j=k+1}^{r}\left[\left(\bar{\tau}_{k}-\bar{\tau}_{j}\right) \int_{t-\bar{\tau}_{k}}^{t-\bar{\tau}_{j}} \dot{\mathbf{z}}^{\text{T}}(t+\theta) \left(\mathbf{I} \otimes \mathbf{S}_{k j} \right) \dot{\mathbf{z}}(t+\theta) d \theta \right]\\
=\ & \dot{\mathbf{z}}^{\text{T}}(t) \sum_{k=1}^{r-1} \sum_{j=k+1}^{r}\left[\left(\bar{\tau}_{k}-\bar{\tau}_{j}\right)^2  \left(\mathbf{I} \otimes \mathbf{S}_{k j} \right) \right] \dot{\mathbf{z}}(t)  -\sum_{k=1}^{r-1} \sum_{j=k+1}^{r}\left[\left(\bar{\tau}_{k}-\bar{\tau}_{j}\right) \int_{t-\bar{\tau}_{k}}^{t-\bar{\tau}_{j}} \dot{\mathbf{z}}^{\text{T}}(t+\theta) \left(\mathbf{I} \otimes \mathbf{S}_{k j} \right) \dot{\mathbf{z}}(t+\theta) d \theta \right]
\end{aligned}
\end{equation}

Combining the derivatives of the components in \eqref{lyap_components}, the derivative of the Lyapunov Krasovskii functional is given by
\begin{equation}
\label{lyap_deriv}
\begin{aligned}
\dot{V}\left(\mathbf{z}_{t}, t\right)=\ & 2 \mathbf{z}^{\text{T}}(t) \left(\mathbf{I} \otimes \mathbf{P} \right)\left[\left( \mathbf{I} \otimes \mathbf{A} \right) \mathbf{z}(t)-\sum_{k=1}^{r} \left(\bar{\mathbf{L}}_k \otimes \mathbf{B} \mathbf{K}  \right) \mathbf{z}\left(t-\tau_{k}(t)\right)\right]\\
& + \sum_{k=1}^{r}\left[\mathbf{z}^{\text{T}}(t) \left(\mathbf{I} \otimes \mathbf{Q}_k \right) \mathbf{z}(t)-\left(1-\dot{\tau}_{k}(t)\right) \mathbf{z}^{\text{T}}\left(t-\tau_{k}(t)\right) \left(\mathbf{I} \otimes \mathbf{Q}_k \right) \mathbf{z}\left(t-\tau_{k}(t)\right)\right] \\
& +\dot{\mathbf{z}}^{\text{T}}(t) \mathbf{\Gamma} \dot{\mathbf{z}}(t)-\sum_{k=1}^{r}\left[\bar{\tau}_{k} \int_{t-\bar{\tau}_{k}}^{t} \dot{\mathbf{z}}^{\text{T}}(t+\theta) \left(\mathbf{I} \otimes \mathbf{R}_k \right) \dot{\mathbf{z}}(t+\theta) d \theta\right] \\
& -\sum_{k=1}^{r-1} \sum_{j=k+1}^{r}\left[\left(\bar{\tau}_{k}-\bar{\tau}_{j}\right) \int_{t-\bar{\tau}_{k}}^{t-\bar{\tau}_{j}} \dot{\mathbf{z}}^{\text{T}}(t+\theta) \left(\mathbf{I} \otimes \mathbf{S}_{k j} \right) \dot{\mathbf{z}}(t+\theta) d \theta \right]
\end{aligned}
\end{equation}.

Since the delays and the delay derivatives are assumed to be bounded such that $\tau_{k}(t) \leq \bar{\tau}_{k}$ for all $t$ and $k$, $\dot{\tau}_{k}(t) \leq \mu_{k}<1$, for all $k$, \eqref{lyap_deriv} can be written as
\begin{equation}
\label{lyap_deriv2}
\begin{aligned}
\dot{V}\left(\mathbf{z}_{t}, t\right) \leq \ & 2 \mathbf{z}^{\text{T}}(t) \left(\mathbf{I} \otimes \mathbf{P} \right)\left[\left( \mathbf{I} \otimes \mathbf{A} \right) \mathbf{z}(t)-\sum_{k=1}^{r} \left(\bar{\mathbf{L}}_k \otimes \mathbf{B} \mathbf{K}  \right) \mathbf{z}\left(t-\tau_{k}(t)\right)\right]\\
& + \sum_{k=1}^{r}\left[\mathbf{z}^{\text{T}}(t) \left(\mathbf{I} \otimes \mathbf{Q}_k \right) \mathbf{z}(t)-\left(1-\mu_{k}\right) \mathbf{z}^{\text{T}}\left(t-\tau_{k}(t)\right) \left(\mathbf{I} \otimes \mathbf{Q}_k \right) \mathbf{z}\left(t-\tau_{k}(t)\right)\right] \\
& +\dot{\mathbf{z}}^{\text{T}}(t) \mathbf{\Gamma} \dot{\mathbf{z}}(t)-\sum_{k=1}^{r}\left[\bar{\tau}_{k} \int_{t-\bar{\tau}_{k}}^{t} \dot{\mathbf{z}}^{\text{T}}(t+\theta) \left(\mathbf{I} \otimes \mathbf{R}_k \right) \dot{\mathbf{z}}(t+\theta) d \theta\right] \\
& -\sum_{k=1}^{r-1} \sum_{j=k+1}^{r}\left[\left(\tau_{k}(t)-\tau_{j}(t)\right) \int_{t-\tau_{k}(t)}^{t-\tau_{j}(t)} \dot{\mathbf{z}}^{\text{T}}(t+\theta) \left(\mathbf{I} \otimes \mathbf{S}_{k j} \right) \dot{\mathbf{z}}(t+\theta) d \theta \right]
\end{aligned}
\end{equation}

Now, using Lemma \ref{lem_S} with $\mathbb{X}$=$\mathbf{I} \otimes \mathbf{S}_{k j}$, one can obtain the following integral inequality
\begin{equation}
\label{eqn_lem_S}
-\left(\tau_{k}(t)-\tau_{j}(t)\right) \int_{t-\tau_{k}(t)}^{t-\tau_{j}(t)} \dot{\mathbf{z}}^{\text{T}}(t+\theta) \left(\mathbf{I} \otimes \mathbf{S}_{k j} \right) \dot{\mathbf{z}}(t+\theta) d \theta \ \leq \ \begin{bmatrix}
\mathbf{z}\left(t-\tau_{j}(t)\right)\\
\mathbf{z}\left(t-\tau_{k}(t)\right) 
\end{bmatrix}^{\text{T}} \begin{bmatrix}
-\left(\mathbf{I} \otimes \mathbf{S}_{k j} \right) & \left(\mathbf{I} \otimes \mathbf{S}_{k j} \right) \\
* & -\left(\mathbf{I} \otimes \mathbf{S}_{k j} \right)
\end{bmatrix} \begin{bmatrix}
\mathbf{z}\left(t-\tau_{j}(t)\right)\\
\mathbf{z}\left(t-\tau_{k}(t)\right) 
\end{bmatrix}
\end{equation}

Further, using Lemma \ref{lem_R} with $\mathbb{Y} =\mathbf{I} \otimes \mathbf{R}_k$, we can write
\begin{equation}
\label{eqn_lem_R}
-\bar{\tau}_{k} \int_{t-\bar{\tau}_{k}}^{t} \dot{\mathbf{z}}^{\text{T}}(t+\theta) \left(\mathbf{I} \otimes \mathbf{R}_k \right) \dot{\mathbf{z}}(t+\theta) d \theta \ \leq \ \begin{bmatrix}
\mathbf{z}(t) \\
\mathbf{z}\left(t-\tau_{k}(t)\right) \\
\mathbf{z}\left(t-\bar{\tau}_{k}\right)
\end{bmatrix}  \begin{bmatrix}
-\left(\mathbf{I} \otimes \mathbf{R}_k \right) & \left(\mathbf{I} \otimes \mathbf{R}_k \right) & \mathbf{0} \\
* & -2 \left(\mathbf{I} \otimes \mathbf{R}_k \right) & \left(\mathbf{I} \otimes \mathbf{R}_k \right) \\
* & * & -\left(\mathbf{I} \otimes \mathbf{R}_k \right)
\end{bmatrix} \begin{bmatrix}
\mathbf{z}(t) \\
\mathbf{z}\left(t-\tau_{k}(t)\right) \\
\mathbf{z}\left(t-\bar{\tau}_{k}\right)
\end{bmatrix} 
\end{equation}

After substituting \eqref{eqn_lem_S} and \eqref{eqn_lem_R} in \eqref{lyap_deriv2} and with some algebraic manipulations, one can express the derivative of the Lyapunov Krasovskii functional as follows
\begin{equation}
\dot{V}\left(\mathbf{z}_{t}, t\right) \leq \Sigma^{T}(t)\left[\mathbf{\Pi}+\bm{\xi}^{T} \mathbf{\Gamma} \bm{\xi}\right] \Sigma(t)
\end{equation}
where $\Sigma(t)= \text{col} \left(\mathbf{z}(t), \ \mathbf{z}\left(t-\tau_{1}(t)\right), \ \mathbf{z} \left(t-\bar{\tau}_{1}\right) , \ \cdots \ , \ \mathbf{z}\left(t-\tau_{r}(t)\right) , \ \mathbf{z}\left(t-\bar{\tau}_{r}\right) \right)$ is the augmented error vector, and $\bm{\xi}$ and $\mathbf{\Pi}$ are as defined in \eqref{gamma_mat} and \eqref{pi_mat}, respectively. Using Lyapunov Krasovskii stability theorem, we can conclude that the system \eqref{z_dynamics} achieves asymptotic stability for $\tau_k \leq \bar{\tau}_k$, for all $k=1, \dots, r$ if the inequality $\left[\mathbf{\Pi}+\bm{\xi}^{T} \mathbf{\Gamma} \bm{\xi}\right]<\mathbf{0}$ holds. Further, using the Schur compliment, one can obtain the LMI in \eqref{lmi_delay}. Moreover, using Lemma \ref{lem_z_x}, the multi-agent system in \eqref{global_state_dyn_wdelay} achieves consensus asymptotically. This concludes the proof.
\end{proof}

\section{Simulation Results}
\label{sec:results}

To demonstrate the preceeding analysis, we consider a multi-agent system with following system matrices:
\begin{equation}
\mathbf{A} = \begin{bmatrix}
-2 & 2\\
-1 & 1
\end{bmatrix}, \qquad \mathbf{B} = \begin{bmatrix}
1\\
0
\end{bmatrix}.
\label{system_dyn_simulation}
\end{equation}

The choice of $\mathbf{A}$ and $\mathbf{B}$ satisfies Assumption \ref{assum_stab}. Let us now choose a stabilizing feedback gain, $\mathbf{K} = \begin{bmatrix}
-2 & -0.5
\end{bmatrix}$ such that $\mathbf{A}-\mathbf{B} \mathbf{K}$ is Hurwitz. We consider a network of 3 agents with the following graph Laplacian matrix,
\begin{equation}
\mathbf{L} = \begin{bmatrix}
0 & 0 & 0\\
-1 & 2 & -1\\
0 & -1 & 1
\end{bmatrix}.
\label{laplacian_simulation}
\end{equation} 

For simulation, the maximum bounds on delay derivatives are taken to be: $\mu_1 = 0.7$, $\mu_2=0.8$ and $\mu_3 = 0.9$. Moreover, upon solving the LMI in \eqref{lmi_delay}, we obtain maximum delay bounds to be $\bar{\tau}_1 = 0.29, \quad \bar{\tau}_2 = 0.18, \quad \bar{\tau}_3 = 0.18$ for consensus among the three agents. Figure \ref{fig:state_trajectory_delay} shows the states of agents with delays of $\tau_1 = 0.25$ s, $\tau_2 = 0.16$ s and $\tau_3 = 0.16$ s. Clearly, all the agents achieve stable consensus in the presence of nonuniform delays within the delay bounds. 

\begin{figure}[thpb]
\centering
      \includegraphics[width = 4.5in]{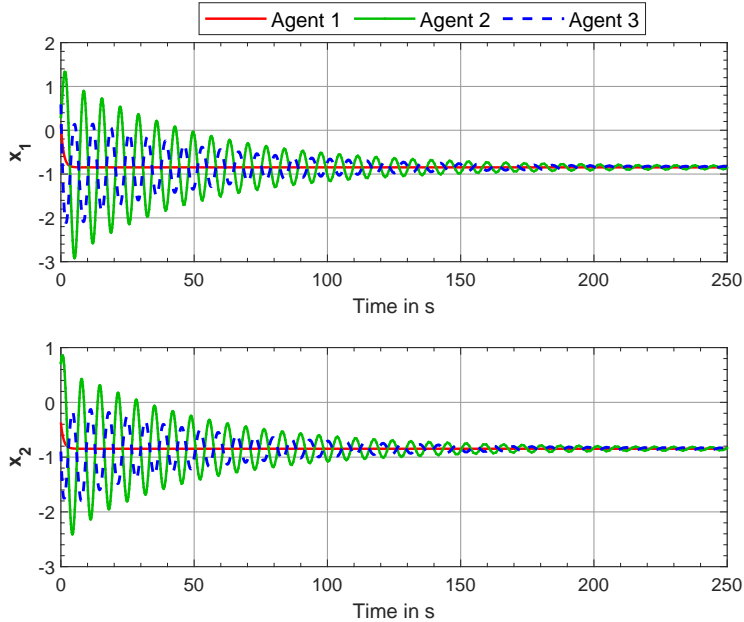}
      \caption{State trajectories of agents with $\tau_1 = 0.25$, $\tau_2 = 0.16$ and $\tau_3 = 0.16$}
	  \label{fig:state_trajectory_delay}
 \end{figure} 
 
 \section{Conclusion}
 \label{sec:conclusion}
 In this paper, we have studied the consensus condition for high-order linear multi-agent system with nonuniform, time-varying communication delay among the agents. First, a fully distributed control protocol was designed for the multi-agent system with delay and subsequently, the problem of state consensus among the agents was transformed to an equivalent problem of stability of the consensus error. Then, using Lyapunov Krasovskii approach, we derive LMI-based approach to characterize the delay margin for multi-agent system with  non-uniform time-varying delays. Numerical results for agents with linear high-order dynamics demonstrate the accuracy of the proposed approach. 
 
 \section*{Appendix}
 \subsection*{Single indexing for communication delays}
 In order to compute the single index $k$ for time-delay $\tau_{ij}$ with double indices, following procedure is used.
 
 \begin{algorithm}
\caption{Computation of $k$} 
\label{single_indexing_delay}
  \begin{algorithmic}[1]
   \STATE \textbf{Initialize:} $k=0$
    \FOR{$i=\{1,2, \ldots, N\}$}
      \FOR{$ j \in \mathcal{N}_{i}$}
       \STATE $k = k+1$
        \STATE $\tau_{k}=\tau_{i j}$
      \ENDFOR
    \ENDFOR
  \end{algorithmic}
\end{algorithm}

\subsection*{Proof to Lemma \ref{lem_z_x}}
\begin{proof}
Since  $\mathbf{z}_i = \mathbf{x}_1-\mathbf{x}_{i+1}$, ($i=1, \cdots, N-1)$. $\mathbf{z} = \text{col} (\mathbf{z}_1, \mathbf{z}_2, \dots, \mathbf{z}_N)$, and $\mathbf{U} = \begin{bmatrix}
\mathbf{1}_{N-1} & -\mathbf{I}_{N-1} \end{bmatrix}$, one can write
\begin{equation}
\mathbf{z}(t) = \left(\mathbf{U} \otimes \mathbf{I}_{n}\right) \mathbf{x}(t). \label{z_rel_x}
\end{equation}

Also with $\mathbf{W} = \begin{bmatrix}
\mathbf{0}^{\text{T}}_{N-1} \\ -\mathbf{I}_{N-1} 
\end{bmatrix}$, we have $\mathbf{U} \mathbf{W} = \mathbf{I}_{N-1}$; therefore, the global state vector of the multi-agent system in \eqref{global_state_dyn_wdelay} can be written as $\mathbf{x}(t) = \left( \mathbf{W} \otimes \mathbf{I}_n \right) \mathbf{z}(t)$. Differentiating \eqref{z_rel_x} with respect to time, we obtain
\begin{equation}
\begin{aligned}
\dot{\mathbf{z}}(t)= \ & \left(\mathbf{U} \otimes \mathbf{I}_{n}\right) \left[\left(\mathbf{I}_{N} \otimes \mathbf{A}\right) \mathbf{x}(t)-\sum_{k=1}^{r}\left(\mathbf{L}_{k} \otimes \mathbf{B} \mathbf{K}\right) \mathbf{x}\left(t-\tau_{k}\right)\right]\\
=\ & \left(\mathbf{U} \otimes \mathbf{I}_{n}\right) \left[\left(\mathbf{I}_{N} \otimes \mathbf{A}\right) \left(\mathbf{W} \otimes \mathbf{I}_{n}\right) \mathbf{z}(t)-\sum^r_{k=1}\left(\mathbf{L}_{k} \otimes \mathbf{B} \mathbf{K}\right) \left(\mathbf{W} \otimes \mathbf{I}_{n}\right) \mathbf{z}(t-\tau_k) \right]\\
=\ & \left(\mathbf{U} \mathbf{W} \otimes \mathbf{A} \right) \mathbf{z}(t) - \sum^r_{k=1}\left[ \left(\mathbf{U} \mathbf{L}_{k} \mathbf{W} \right) \otimes \mathbf{B} \mathbf{K}\right] \mathbf{z}(t-\tau_k)\\
=\ & \left(\mathbf{I}_{N-1} \otimes \mathbf{A} \right) \mathbf{z}(t) - \sum_{k=1}^{r} \left(\bar{\mathbf{L}}_k \otimes \mathbf{B} \mathbf{K}  \right) \mathbf{z}\left(t-\tau_{k}(t)\right) 
\end{aligned}
\end{equation}

Note that, as $\mathbf{z}(t) \to 0$, $\mathbf{x}_i \to \mathbf{x}_1$, $i=2, \dots, N$. This completes the proof.
\end{proof}

\section*{Acknowledgments}
This work was supported by the Office of Naval Research (grant number N00014-18-1-2215).

\bibliography{Bhusal_Subbarao_Delay_2020}
\end{document}